\documentclass[
	12pt, 
	a4paper, 
	]{article}

\usepackage{amssymb} 
\usepackage{amsmath} 
\usepackage{amsthm}
\usepackage{geometry} 
\usepackage{graphicx} 
\usepackage{caption,subcaption} 
\usepackage{setspace} 
\usepackage{natbib} 
\usepackage{array} 
\usepackage[breaklinks,colorlinks,linkcolor=blue!70!black,filecolor=blue!70!black,urlcolor=blue!70!black,citecolor=blue!70!black]{hyperref}  
\usepackage[nameinlink]{cleveref} 

\usepackage{enumitem} 

\usepackage{lineno} 

\usepackage{bm,dsfont,sidecap,comment,sgame}

\usepackage{tikz} 
\usetikzlibrary{arrows,positioning,calc}
\usetikzlibrary{shapes.geometric}
\usepackage{pgfplots}
\pgfplotsset{compat=1.3}

\newtheorem{corollary}{Corollary}
\newtheorem{lemma}{Lemma}

\newtheorem{proposition}{Proposition}

\newtheorem{fact}{Fact}

\newtheorem{set-up}{Set-up}

\newtheorem*{theorem*}{Theorem}
\newtheorem*{corollary*}{Corollary}
\newtheorem*{lemma*}{Lemma}
\newtheorem*{observation*}{Observation}
\newtheorem*{proposition*}{Proposition}
\newtheorem*{claim*}{Claim}
\newtheorem*{fact*}{Fact}
\newtheorem*{assumption*}{Assumption}
\newtheorem*{assumptionA*}{Assumption A}
\newtheorem*{set-up*}{Set-up}

\theoremstyle{definition}

\newtheorem*{definition*}{Definition}

\newtheorem*{problem*}{Problem}
\newtheorem{example}{Example}
\newtheorem*{example*}{Example}

\usepackage{apptools}
\AtAppendix{\counterwithin{lemma}{section}}
\AtAppendix{\counterwithin{proposition}{section}}

\renewenvironment{proof}[1][\proofname]{{\noindent\textbf{#1.}}}{\qed \vspace{\topsep}}

\newcolumntype{L}[1]{>{\raggedright\let\newline\\arraybackslash\hspace{0pt}}m{#1}}
\newcolumntype{C}[1]{>{\centering\let\newline\\arraybackslash\hspace{0pt}}m{#1}}
\newcolumntype{R}[1]{>{\raggedleft\let\newline\\arraybackslash\hspace{0pt}}m{#1}}

\newcommand{\bolds}{\boldsymbol{s}}
\newcommand{\boldv}{\boldsymbol{v}}

\newcommand{\boldsigma}{\boldsymbol{\sigma}}

\newcommand{\boldc}{\boldsymbol{c}}
\newcommand{\boldq}{\boldsymbol{q}}

\newcommand{\st}{\quad \text{subject to} \quad}

\begin{document}

\begin{titlepage}
\singlespacing

\title{Crowdsearch\thanks{We thank Rainer B\"{o}hme, Tapas Kundu, Priel Levy, and conference participants at the Conference on Economic Design (2023), the Annual Conference on Contests (2023), the Stony Brook International Conference on Game Theory (2023), and ETH Z\"{u}rich for their comments. A one-page abstract of an earlier version of this paper, titled ``Decentralized Attack Search and the Design of Bug Bounty Schemes,''  appears in the Proceedings of the 16th International Symposium on Algorithmic Game Theory, SAGT 2023.  This research was partially supported by the Zurich Information Security and Privacy Center (ZISC). Otto Schmidt and Pio Blieske provided excellent research assistance. All errors are our own.}}

\author{
	Hans Gersbach\\
	\normalsize KOF Swiss Economic Institute, \\
        \normalsize ETH Zurich, and CEPR\\ 
	\normalsize Leonhardstrasse 21\\
	\normalsize 8092 Zurich, Switzerland\\ 
	\normalsize \href{mailto:hgersbach@ethz.ch}{hgersbach@ethz.ch}
	\and
	Akaki Mamageishvili\\
	\normalsize Offchain Labs\\ 
	\normalsize Zurich, Switzerland\\ 
	\normalsize \href{mailto:amamageishvili@offchainlabs.com}{amamageishvili@offchainlabs.com}
	\and
	Fikri Pitsuwan\\
    \normalsize Jeb E. Brooks School of Public Policy \\
        \normalsize Cornell University\\ 
	\normalsize Ithaca, New York 14853, USA\\
	\normalsize \href{mailto:fp67@cornell.edu}{fp67@cornell.edu}
	}

\date{Last updated: \today}

\maketitle
\begin{abstract}

\noindent A common economic process is \textit{crowdsearch}, wherein a group of agents is invited to search for a valuable object, such as a software vulnerability or a proof-of-work nonce. We study a binary model of crowdsearch in which agents with heterogeneous search costs decide whether or not to search. Inviting more agents directly expands the pool of potential searchers, but it also increases competition for the prize and lowers individual search incentives. This crowding-out effect may dominate, so that inviting more agents can lower the probability of finding the object. We show that inviting an unlimited crowd need not guarantee success. We also provide sufficient conditions, based on the lower-tail behavior of the cost distribution, under which the probability of success is eventually increasing or eventually decreasing in the number of agents.
    \\
	\vspace{0in}
	\vspace{0in}
	\bigskip
\end{abstract}

\thispagestyle{empty}
\end{titlepage}

\setcounter{page}{2}


\onehalfspacing

    
\section{Introduction} \label{sec:introduction}

A common economic process is \textit{crowdsearch}, wherein a group of agents is invited to search for a physical or virtual object that is valuable to a principal, a community of principals, or an entire society.

An example is crowdsourced security, or \textit{bug bounty}, where a crowd searches for vulnerabilities, or bugs, in software and blockchain infrastructure in exchange for rewards, or bounties.\footnote{For literature on bug bounty programs, see \citet{bohme2006, malladi2020, walshe2020, zrahia2022,akgul2023} and references therein.} This type of contest has become a major tool for detecting vulnerabilities in software used by governments, tech companies, and blockchains. Bug bounty programs are particularly important for blockchain infrastructure providers, since many such projects do not have dedicated security teams testing software upgrades. Once software is deployed on a blockchain, there may be limited ability to reverse course and limited legal recourse against system exploitation, at least until the next hard fork\textemdash a major change in the blockchain protocol \citep{breidenbach2018,bohme2020}. Another example is the proof-of-work (PoW) protocol in Bitcoin and many other cryptocurrency networks. PoW provides incentives for agents to verify, broadcast, and reach consensus on a set of transactions. In this case, the crowd, namely miners, searches for a virtual object\textemdash a number called the \textit{nonce}\textemdash and the first miner to find it obtains the right to add a block of processed transactions and receives tokens as a reward \citep{dimitri2017, arnosti2022, ma2019,leshno2020, halaburda2022}.

In this paper, we provide a simple game-theoretic model of \textit{crowdsearch}, abstracting from the details of particular applications. We focus on a stylized setting in which a crowd of arbitrary size is invited to search for an object that is valuable to the principal. Each agent makes a binary decision whether or not to search at a heterogeneous private cost, which reflects the agent's ability. If the object is found, the principal offers a prize to one of the agents who find it.

Building on the important work of \citet{ghosh2016} and \citet{sarne2017}, we model crowdsearch as a simple contest, where each agent makes a binary costly decision whether or not to search. This allows us to examine how the size of the invited crowd affects the probability of success and whether inviting more agents benefits or harms the search. Our justification for modeling the search decision as binary is twofold. First, decisions to enter a search may entail substantial fixed costs. Second, the binary decision can be interpreted as a costly entry decision into a contest in which there is no additional strategic action after entry. Conditional on entry, agents have the same probability of finding the object and are treated symmetrically in the allocation of the prize. More generally, the simple contest assumption allows us to focus on the extensive margin of participation.

An important feature of crowdsearch that differentiates it from a standard crowdsourcing contest is the principal's objective. Whereas the goal may be to elicit the best entry in a crowdsourcing contest, the goal of crowdsearch is to maximize the likelihood of finding the object, given a budget or prize that can be paid. We call this objective the \textit{probability of success}. In particular, we focus on how the size of the invited crowd affects the probability of success. What is the effect of having a larger crowd? Can inviting more agents ever reduce the probability that the object is found? What happens to the equilibrium threshold and the probability of success as the crowd becomes large?

We obtain the following results. First, we show that the probability of success can be non-monotonic in the number of agents. Having more agents directly expands the pool of potential searchers, but it also increases competition for the prize and lowers individual search incentives. We identify two forces that determine which effect dominates: the level of equilibrium participation and the responsiveness of participation to the equilibrium threshold. We also show that when participation becomes less responsive as the threshold falls, any decline in the probability of success can occur only over an initial range of crowd sizes.

Second, we characterize the role of the difficulty parameter. When the object is difficult to find, an additional searching agent generates little additional competition, so the probability of success is increasing in the number of agents for any prize. When the object is easier to find, the crowding-out effect can dominate for some prize. For the family $F(c)=c^\alpha$, we characterize the cutoff explicitly and show how it varies with $\alpha$.

Third, we study large contests. We show that the equilibrium threshold converges to the lower bound of the cost distribution. If costs can be arbitrarily close to zero, the probability of success converges to one. If costs are bounded away from zero, however, inviting an unlimited crowd need not guarantee success. Finally, we provide sufficient conditions, stated in terms of the lower-tail behavior of $cf(c)$, under which the probability of success is eventually increasing or eventually decreasing.

The rest of the paper is organized as follows. We review the related literature in \Cref{sec:literature} and introduce the model in \Cref{sec:model}. In \Cref{sec:analysis}, we characterize the equilibrium and study the probability of success. \Cref{sec:largecontest} examines the asymptotic behavior of the model as the crowd becomes large. \Cref{sec:conclusion} concludes. The appendix contains all proofs.



\section{Literature} \label{sec:literature}

Our paper belongs to the large literature on contests in economics and computer science, where agents compete by exerting costly effort to win prizes \citep{konrad2009,vojnovic2016,segev2020}. In traditional models \citep{tullock1980,lazear1981}, agents make strategic choices on the \textit{amount} of effort to exert and the design objective is typically to elicit the highest sum of efforts. This can be done by appropriately splitting up the prize \citep{moldovanu2001}, sequencing the timing of choices \citep{hinnosaar2024}, or by developing a revelation mechanism to select a subset of contestants from a pool of candidates \citep{mercier2018}.

In contrast, the design objective in crowdsourcing contests or research tournaments is to elicit the highest quality of submissions. Pioneering this literature, \citet{taylor1995} studies a contest where invited agents first pay an entry cost into the tournament. The agents then decide whether or not to draw costly innovation from a common distribution. The winner is the agent with the highest value of innovation. Our contest model can be thought of as focusing on the entry stage of \citet{taylor1995}'s model, with equilibrium play occurring in the contest stage afterward. The main departure is that the entry cost in our model is private and heterogeneous. \citet{dipalantino2009} consider a game where agents with private skill levels select and compete in multiple simultaneous contests with different rewards, modeled as all-pay auctions. They show that in equilibrium, agents partition themselves by skill levels\textemdash high-skilled agents compete in the contest with high rewards. \citet{archak2009} study a single contest modeled as an all-pay auction, but with multiple rewards. They focus on the asymptotic behavior and show that as the number of contestants goes to infinity, only the highest-ability, i.e., lowest-cost, agents matter, implying that equilibrium behavior in a large contest is distribution-free in the sense that only the support matters, and not the specific shape of the distribution. We similarly find that the large-crowd limit is governed by the lowest-cost agents. At the same time, the way the probability of success approaches its limit depends on the lower-tail behavior of the cost distribution, summarized by $cf(c)$.

Our paper complements the mainstream contest literature by studying contests with ``simple agents'' as pioneered by \citet{ghosh2016}. Our model differs from that of \citet{ghosh2016} and follow-up work such as \citet{levy2017} and \citet{sarne2017} in a major aspect. The quality of submission\textemdash which is private information in \citet{ghosh2016} and drawn after the decision is made in \citet{levy2017}\textemdash is absent in our model. Instead, we model the cost of participation as private information, which only affects the agents' strategic behavior and not the principal's objective. Albeit in a very different setting, the model most closely related to ours is \citet{ghosh2013}, where agents with different privacy requirements decide whether or not to participate in a database. With more agents, the database becomes more private. \citet{ghosh2013} show that the symmetric equilibrium has a threshold structure in which agents with less privacy requirement participate. In a sense, \citet{ghosh2013} study the model opposite to ours where participation induces further participation. This leads to some noteworthy differences. First, whereas the equilibrium threshold in our model is characterized by a fixed point of a strictly decreasing function, the equilibrium thresholds in their model are fixed points of a strictly increasing function. Second, the asymptotic number of participants in our model depends crucially on the support of the distribution of private information, while in \citet{ghosh2013}'s model this number always diverges to infinity.

Importantly, our model is related to models of contests and auctions with entry. While \citet{liu2019a} study optimal prizes in contests with entry, our model is more closely related to auctions with entry (\citet{levin1994} and \citet{moreno2011}) as the prize allocation after entry in our model is probabilistic and does not depend on further effort decision. In \citet{levin1994} and \citet{moreno2011}, agents decide whether or not to enter an auction at a cost, and draw their valuation after entry. \citet{levin1994} study the symmetric mixed-strategy entry equilibrium in a standard auction with common entry costs. \citet{moreno2011} consider a version with heterogeneous entry costs. Similar to ours, entry decisions are characterized by a threshold. The economic differences to \citet{moreno2011} are as follows. First, in their auction model, agents draw their own value of the item in the auction. In our model, however, agents can be thought of as drawing the value of the designer: either the object is found, which gives the principal positive value, or it is not. Second, the prize in our model affects equilibrium participation and therefore the probability of success, while the reserve price in an auction affects allocation and revenue. Since we have a simple model of contests, we are able to derive precise conditions for how the number of agents affects the probability of success, how this effect depends on the difficulty of finding the object, and how the lower tail of the cost distribution determines the asymptotic behavior of the search.


\section{Model} \label{sec:model}

A principal invites a set of agents $N = \{1,\dots,n\}$, $n \geq 2$, to search for an object in exchange for a prize. If agent $i$ decides to search ($s_i = 1$), $i$ incurs a cost $c_i$ and finds the object independently with probability $q \in (0,1]$, where $1-q$ is interpreted as the \textit{difficulty} of the search. Otherwise, agent $i$ decides not to search ($s_i = 0$) and does not find the object. Agent $i$'s cost, $c_i$, is private information. Costs are drawn independently and identically from an absolutely continuous distribution $F$ with support $[\underline{c},\overline{c}]$, where $0\leq\underline{c}<\overline{c}\leq\infty$. We assume that the density $f$ is continuous and strictly positive on $(\underline{c},\overline{c})$. 
If the object is found, a prize $V > 0$ is given uniformly randomly to one of the agents who found it.\footnote{This is a \textit{winner-takes-all} contest. The prize $V$ could also be interpreted as an expected reward, i.e., the monetary prize multiplied by the probability that the object exists.
} We assume that $\underline{c} <  qV$ and $V \frac{1-(1-q)^2}{2} < \overline{c}$.\footnote{The former parameter restriction ensures that an agent has an incentive to search even if there are no other agents, while the latter ensures that there are agents that do not search.} 

We write $\bolds = (s_1,\dots,s_n)$ and $\bolds_{-i} = (s_1,\dots,s_{i-1},s_{i+1},\dots,s_n)$, and let $S = \{j \in N: s_j = 1\}$ and $S_{-i} = \{j \in N\setminus \{i\}: s_j = 1\}$ denote the set of agents who search and the set of agents other than agent $i$ who search, respectively. The payoff of agent $i$ is
\begin{equation} 
u_i(s_i, \bolds_{-i},c_i) = s_i\left( p(\bolds_{-i}) V - c_i \right),
\end{equation}
where 
\begin{equation} \label{eq:payoffprob}
p(\bolds_{-i}) \equiv q \sum_{t=0}^{|S_{-i}|} \binom{|S_{-i}|}{t} \frac{q^t(1-q)^{|S_{-i}| - t}}{t+1}
\end{equation}
is the probability that agent $i$ finds the object and wins the prize, conditional on searching. To win the prize, the agent must find the object\textemdash captured by the first $q$ in~\eqref{eq:payoffprob}\textemdash and be chosen among all those who found it\textemdash the sum with binomial coefficients in~\eqref{eq:payoffprob}.

A strategy profile is denoted $\boldsigma = (\sigma_1,\dots,\sigma_n)$, where a strategy $\sigma_i: [\underline{c},\overline{c}] \rightarrow \{0,1\}$ maps an agent's private information to a search decision. A strategy profile $\boldsigma^*$ is a Bayesian Nash equilibrium if for all $i$, $c$, and $s_i$, \[\mathbb{E}[u_i(\sigma^*_i(c_i),\boldsigma^*_{-i}(\boldc_{-i}),c_i)|c_i = c] \geq \mathbb{E}[u_i(s_i,\boldsigma^*_{-i}(\boldc_{-i}),c_i)|c_i = c].\]

In the main analysis, we characterize the Bayesian Nash equilibrium in a \textit{symmetric threshold strategy}, where all agents use the same strategy of the form:
\[
\sigma(c_i) = \left\{ \begin{array}{rcl}
 1 &  \mbox{if} & c_i \leq \hat{c} \\ 
 0 & \mbox{if} & c_i > \hat{c} \\
\end{array}\right .
\]
for some threshold $\hat{c}$.\footnote{Note that restricting attention to a threshold strategy is without loss of generality since all equilibria are in threshold strategies. However, restricting attention to symmetric strategies entails some loss of generality. See the extended version of this paper \citep{gersbach2023c} for details. 
} Given a symmetric threshold strategy with threshold $\hat{c}$, the ex-ante probability that the object is found is given by 
\[
   P(\hat{c},q,n) \equiv  1 - (1-qF(\hat{c}))^n ,
\]
which we shall call the \textit{probability of success}. We focus on the equilibrium probability of success induced by the prize $V$ and the number of invited agents $n$.


\section{Equilibrium Analysis and Probability of Success} \label{sec:analysis}

This section analyzes the game and discusses comparative statics for the probability of success. Our first result characterizes a simple fixed-point condition for the unique symmetric equilibrium threshold.

\begin{proposition} \label{prop:eqm}
The unique symmetric equilibrium threshold $c^* \equiv c^*(V,q,n)$ is the solution to
\begin{equation} \label{eq:symeqm}
\hat{c} = V \Phi(\hat{c},q,n),
\end{equation}
where
\begin{equation} \label{eq:phi}
     \Phi(\hat{c},q,n) \equiv \frac{P(\hat{c},q,n)}{nF(\hat{c})} = \frac{1-(1-qF(\hat{c}))^n}{nF(\hat{c})}.
\end{equation}
The equilibrium threshold
$c^*(V,q,n)$ is (i) increasing in $V$, (ii) increasing in $q$, and (iii) decreasing in $n$.
\end{proposition}

We henceforth refer to $c^*$ simply as the \textit{equilibrium}.\footnote{Technically, the equilibrium is \textit{generically} unique in the sense that one could alternatively define a threshold strategy that specifies $0$ (no search), or even randomizes, at the threshold value.} To ease exposition, we suppress explicit dependence of $c^*$ and $\Phi$ on $V$, $q$, and $n$, when appropriate. Condition~\eqref{eq:symeqm} is an indifference condition capturing the fact that in an equilibrium, an agent of type $c^*$ must be indifferent between searching and not searching. The left-hand side is the cost of the search and the right-hand side is the expected reward: $V$ times $\Phi$. Indeed, $\Phi$ is the probability that agent $i$ wins, given that all other agents use the same threshold strategy. As defined in~\eqref{eq:phi}, $\Phi$ has an intuitive interpretation: it is the probability that the object is found, divided by the expected number of agents who search. The reason is that if the object is found at all, then the agents participating in the search have the same chance to obtain the reward.

The comparative statics results are intuitive. If the prize $V$ is increased, agents have more incentive to search. Agents with higher costs will now search when they otherwise would not. The same is true when $q$, i.e., the probability that the object is found conditional on search, increases. Even though other agents are more likely to find the object as well, the direct effect of an agent's increase in the probability of finding the object dominates. Lastly, more agents intensify competition, which lowers the probability that an agent wins the prize, thus lowering the equilibrium threshold.

We now investigate the probability of success in equilibrium, $P(c^*(V,q,n),q,n)$, which we denote by $P^*(V,q,n)$. It is straightforward to see that $P^*(V,q,n)$ increases with $V$ and $q$. The comparative statics with respect to $n$, however, is more interesting and forms the focus of the remainder of the paper. Unless otherwise stated, we therefore keep $V$ and $q$ fixed and denote the equilibrium threshold and the equilibrium probability of success when $n$ agents are invited to search by $c_n\equiv c^*(V,q,n)$ and $P_n \equiv P^*(V,q,n)$, respectively. We show that $P_n$ may decrease or increase with $n$.\footnote{This is akin to the group size paradox, a phenomenon well-known in the context of organizations, where \citet{olson1965} made the argument that small groups are more effective in providing collective goods for their members due to the free-riding problem, and a large literature has developed to examine the validity and limits of this argument (see, e.g., \citet{esteban2001} for an assessment). In our context, there is no collective good provision, and the crowdsearch designer simply wants to maximize the probability of finding the object and may achieve this by limiting the number of agents who can participate.} Intuition suggests that the probability of success should go up with the number of agents. However, as we have seen, more agents result in heightened competition, which lowers the participation threshold. That is, agents crowd out each other's individual incentives to search. Either force may dominate, depending on the specifications of the cost distribution and the parameters of the model.

There are two channels through which the crowding-out effect can dominate when $n$ increases. The first operates through the level of equilibrium participation, $F(c_n)$. When participation is high, the crowding-out effect is more likely to dominate since competition for the prize is intense. A high reward can generate this effect even under a uniform cost distribution, as illustrated by the following example.

\begin{example} \label{ex:1}
Consider $F(c) = c$ for $0 \leq c \leq 1$, and let $V = 1.999$ and $q=1$. \Cref{fig:example1} shows the numerical values of $c_n$ and $P_n$. The equilibrium thresholds $c_n$ are decreasing in $n$, while $P_n$ decreases from $n=2$ to $n=4$ and increases thereafter.
\qed
\end{example}

\begin{figure}[t]
\centering
\begin{subfigure}[ht]{0.5\textwidth}
\centering
\begin{tikzpicture}[scale=0.7]
\begin{axis}[
    axis y line*=left,
    ymin=0.5, ymax=1,
    xlabel={$n$},
    ylabel={\textcolor{red}{$c_n$}},
    ylabel style={rotate=-90},
    xtick={2,3,4,5,6}
]
\addplot[mark=asterisk,red]
    coordinates {
    (2,0.9998)
    (3,0.8136)
    (4,0.7042)
    (5,0.6301)
    (6,0.5755)
    };
    \label{plot_eqm1}
\end{axis}

\begin{axis}[
  axis y line*=right,
  axis x line=none,
  ymin=0.990, ymax=1,
  yticklabel style={/pgf/number format/.cd,fixed,precision=3},
  ylabel={\textcolor{blue}{$P_n$}},
  ylabel style={rotate=-90}
]
\addlegendimage{/pgfplots/refstyle=plot_eqm1}\addlegendentry{$c_n$}
\addplot[mark=o,blue]
  coordinates{
    (2,0.9999)
    (3,0.9935)
    (4,0.9923)
    (5,0.9931)
    (6,0.9941)
    }; 
    \addlegendentry{$P_n$}
\end{axis}
\end{tikzpicture}
\end{subfigure}
\begin{subfigure}[ht]{0.4\textwidth}
\centering
\vspace{-1em}
\begin{tabular}{ |c|c|c| } 
 \hline
 $n$ & $c_n$ & $P_n$  \\ 
 \hline \hline
 2 & 0.9998 & 0.9999 \\ 
 \hline
 3 & 0.8136 & 0.9935 \\ 
 \hline
 4 & 0.7042 & 0.9923 \\ 
 \hline
 5 & 0.6301 & 0.9931 \\ 
 \hline
 6 & 0.5755 & 0.9941 \\ 
 \hline
\end{tabular}
\end{subfigure}
    \caption{$c_n$ and $P_n$ for \Cref{ex:1}.}
    \label{fig:example1}
\end{figure}

In \Cref{ex:1}, the high reward places the initial equilibrium threshold close to the upper end of the support. Since the distribution is uniform, this also means that $F(c_n)$ is close to one, so almost all agents search. Competition for the prize is therefore intense, and the resulting decline in participation can initially dominate the direct effect of having more potential searchers.

The second channel operates through the marginal mass of agents near the equilibrium threshold. A decrease in $c_n$ lowers participation by excluding agents with costs close to the threshold. This force is captured by $c_nf(c_n)$, which measures the responsiveness of participation to a proportional change in the equilibrium threshold.\footnote{Equivalently, the responsiveness is captured by the elasticity, $\varepsilon(c)\equiv cf(c)/F(c)$.} The following example illustrates this channel.

\begin{example} \label{ex:2}
Consider $F(c)=c^{20}$ for $0\leq c\leq1$, and let $V=1$ and $q=1$. \Cref{fig:example2} shows the numerical values of $c_n$ and $P_n$. The equilibrium threshold decreases with $n$, while $P_n$ decreases from $n=2$ to $n=4$ and increases thereafter.
\qed
\end{example}

In \Cref{ex:2}, $cf(c)=20c^{20}$, which is large around the equilibrium thresholds displayed in the figure. Thus, a decline in the threshold leads to a large decline in participation, which can dominate the direct effect of having more potential searchers.

\begin{figure}[t]
\centering
\begin{subfigure}[ht]{0.5\textwidth}
\centering
\begin{tikzpicture}[scale=0.7]
\begin{axis}[
    axis y line*=left,
    ymin=0.86, ymax=0.92,
    xlabel={$n$},
    ylabel={\textcolor{red}{$c_n$}},
    ylabel style={rotate=-90},
    xtick={2,3,4,5,6}
]
\addplot[mark=asterisk,red]
    coordinates {
    (2,0.9151)
    (3,0.8951)
    (4,0.8828)
    (5,0.8739)
    (6,0.8669)
    };
    \label{plot_eqm2}
\end{axis}

\begin{axis}[
  axis y line*=right,
  axis x line=none,
  ymin=0.285, ymax=0.315,
  ytick={0.29,0.30,0.31},
  ylabel={\textcolor{blue}{$P_n$}},
  ylabel style={rotate=-90}
]
\addlegendimage{/pgfplots/refstyle=plot_eqm2}\addlegendentry{$c_n$}
\addplot[mark=o,blue]
  coordinates{
    (2,0.3106)
    (3,0.2924)
    (4,0.2917)
    (5,0.2948)
    (6,0.2989)
    }; 
    \addlegendentry{$P_n$}
\end{axis}
\end{tikzpicture}
\end{subfigure}
\begin{subfigure}[ht]{0.4\textwidth}
\centering
\vspace{-1em}
\begin{tabular}{ |c|c|c| } 
 \hline
 $n$ & $c_n$ & $P_n$  \\ 
 \hline \hline
 2 & 0.9151 & 0.3106 \\ 
 \hline
 3 & 0.8951 & 0.2924 \\ 
 \hline
 4 & 0.8828 & 0.2917 \\ 
 \hline
 5 & 0.8739 & 0.2948 \\ 
 \hline
 6 & 0.8669 & 0.2989 \\ 
 \hline
\end{tabular}
\end{subfigure}
    \caption{$c_n$ and $P_n$ for \Cref{ex:2}.}
    \label{fig:example2}
\end{figure}

To further investigate the forces at play, we treat $n$ as a continuous variable for the remainder of the analysis. We state our formal results for the continuous extension of $n$, while the numerical examples illustrate the corresponding behavior when $n$ is restricted to integers. 

We first derive a condition, in terms of the cost distribution and the degree of difficulty, under which $P_n$ increases or decreases with $n$. A direct calculation gives\footnote{A detailed derivation is provided in the proof of \Cref{lemma:dpdn}.}
\begin{equation} \label{eq:dpdn}
    \frac{\mathrm{d}P_n}{\mathrm{d} n} = (1-qF(c_n))^n \left[ \frac{nqf(c_n)}{1-qF(c_n)} \frac{\mathrm{d}c_n}{\mathrm{d}n} - \ln{(1-qF(c_n))}\right].
\end{equation}
From~\eqref{eq:dpdn}, we see that $\mathrm{d}P_n/\mathrm{d}n \geq 0$ if and only if the magnitude of $\mathrm{d}c_n/\mathrm{d}n$, which is negative, is not too large. Combining~\eqref{eq:dpdn} with the equilibrium condition~\eqref{eq:symeqm}, we can derive a condition that identifies the range of $n$ over which the probability of success declines.

\begin{lemma} \label{lemma:dpdn}
    $\mathrm{d}P_n/\mathrm{d}n \geq 0$ if and only if
\begin{equation} \label{eq:Pncondition}
L(q,F(c_n)) \geq c_n f(c_n),
\end{equation}
where 
\[L(q,y) \equiv y \frac{-(1-qy)\ln(1-qy)}
    {qy+(1-qy)\ln(1-qy)}.
    \]
\end{lemma}

\Cref{lemma:dpdn} gives a concise condition to help discern the behavior of $P_n$, as it does not depend on how the equilibrium threshold $c_n$ changes. The two examples above illustrate the two sides of~\eqref{eq:Pncondition}. The function $L(q,y)$ is positive and strictly decreasing in $y$. Thus, holding $q$ fixed, higher equilibrium participation $F(c_n)$ makes the left-hand side of~\eqref{eq:Pncondition} smaller. In \Cref{ex:1}, the high reward keeps the equilibrium threshold and participation high, making the left-hand side small. On the other hand, the right-hand side, $c_nf(c_n)$, captures how much participation responds to a proportional change in the equilibrium threshold. In \Cref{ex:2}, the right-hand side is large around the displayed thresholds. In either case,~\eqref{eq:Pncondition} may fail and the crowding-out effect may dominate. An implication is that the principal should pay close attention to both the level of participation and the responsiveness of participation to the equilibrium threshold when trading off the direct effect of a larger crowd against the crowding-out effect.

The condition also yields a general implication for the shape of $P_n$. Since $c_n$ decreases with $n$, the left-hand side of~\eqref{eq:Pncondition} increases with $n$. If $cf(c)$ is weakly increasing in $c$, then the right-hand side is non-increasing in $n$. We have the following result.

\begin{proposition} \label{prop:initialdecline}
Suppose $cf(c)$ is weakly increasing in $c$. If $\mathrm{d}P_n/\mathrm{d}n|_{n=m} \geq 0$ at some $m\geq2$, then $\mathrm{d}P_{n}/\mathrm{d}{n}|_{n=m'}>0$ for every $m'>m$.
\end{proposition}

\Cref{prop:initialdecline} shows that once $P_n$ starts increasing, it keeps increasing for all larger values of $n$. In particular, for $F(c)=c^\alpha$ on $[0,1]$, $\alpha>0$, we have $cf(c)=\alpha c^\alpha$, so the condition holds. Thus, $P_n$ can have at most one turning point, changing from decreasing to increasing. In other words, a decline in $P_n$, if any, must occur over an initial range of $n$, as illustrated in \Cref{ex:1} and \Cref{ex:2}.


Condition~\eqref{eq:Pncondition} highlights the role of the difficulty parameter $q$ in determining whether $P_n$ is monotone in $n$. For a fixed level of participation, $L(q,y)$ is strictly decreasing in $q$, so the condition is easier to satisfy when the object is more difficult to find. Intuitively, when $q$ is low, an additional searching agent is unlikely to find the object, so increasing $n$ generates little additional competition for the prize. By contrast, when the object is easy to find, i.e., when $q$ is high, competition is more intense and the induced decline in search incentives can dominate. The next proposition formalizes this intuition.

\begin{proposition} \label{prop:qdagger}
Suppose $\sup_c cf(c)<\infty$. Then there exists $q^\dagger\in(0,1]$ such that if $q<q^\dagger$, then $P_n$ is increasing in $n$ for all $n\geq2$ and any $V$. If $q>q^\dagger$, then there exists $V$ such that $P_n$ is decreasing for some $n \geq 2$.
\end{proposition}

\Cref{prop:qdagger} gives a cutoff characterization of the role of difficulty. The assumption $\sup_c cf(c)<\infty$ ensures that the cutoff $q^\dagger$ is strictly positive. Below the cutoff, the object is sufficiently difficult to find that the direct effect of adding more potential searchers dominates for any reward and crowd size. Above the cutoff, the object is easy enough to find that, for some reward, the crowding-out effect can dominate.

For $F(c)=c^\alpha$ on $[0,1]$, $\alpha>0$, the cutoff can be characterized explicitly. Since $L(q,c^\alpha)$ is decreasing in $c$ and $cf(c)=\alpha c^\alpha$ is increasing in $c$, condition~\eqref{eq:Pncondition} gets harder to satisfy as $c \rightarrow 1$. Thus, the cutoff, now denoted by $q^\dagger_\alpha$, is the unique solution to $L(q^\dagger_\alpha,1)=\alpha$. Moreover, $q^\dagger_\alpha$ is strictly decreasing in $\alpha$, with $q^\dagger_\alpha\rightarrow1$ as $\alpha\rightarrow0$ and $q^\dagger_\alpha\rightarrow0$ as $\alpha\rightarrow\infty$. In other words, as $\alpha$ increases, a lower value of $q$ is sufficient for some reward to generate a decline in $P_n$.

A remark is in order at this point. \Cref{prop:qdagger} is stated for real values of $n$, so a local decrease does not by itself imply a decrease between integer values of $n$. The two examples above show, however, that such a decline can also occur when $n$ is restricted to integers. In \Cref{ex:1}, where $F(c)=c$, we have $q^\dagger_1\approx0.715$, while in \Cref{ex:2}, where $F(c)=c^{20}$, we have $q^\dagger_{20}\approx0.092$. In both of these examples, $q=1$ is above the relevant cutoff, and the chosen rewards generate a decline in $P_n$ for integer values of $n$.

\section{Large Contests} \label{sec:largecontest}

In this section, we examine the asymptotic behavior of $c_n$ and $P_n$ as the number of agents grows large. We first establish that the equilibrium threshold converges to $\underline{c}$. Thus, as the number of agents grows, individual search incentives decline, and in the limit only types with costs arbitrarily close to the lower bound are willing to search.\footnote{\citet{archak2009} note a similar result in the context of an all-pay auction. \citet{moreno2011} arrive at a similar result for auctions with endogenous entry. They show that an auction model with heterogeneous entry costs is asymptotically equivalent to one with homogeneous cost equal to $\underline{c}$.} The next question of interest is the behavior of aggregate search activity and of the probability of success. We show that aggregate search activity does not vanish. If $\underline{c}=0$, the expected number of agents who search diverges and the probability of success converges to one. If $\underline{c}>0$, the expected number of agents who search converges to a strictly positive constant, and the probability of success converges to a strictly positive limit below one. Formally, we have

\begin{proposition} \label{prop:largen}
The following holds:
\begin{enumerate}
    \item[(i)]  For any $\underline{c} \geq 0$, $c_n \rightarrow \underline{c}$.
    \item[(ii)] If $\underline{c} = 0$, then $nF(c_n) \rightarrow \infty$. If $\underline{c} > 0$, then $nF(c_n) \rightarrow \kappa(\underline{c})$, where the constant $\kappa \equiv \kappa(\underline{c})$ is the unique solution to $\underline{c} = V\frac{1-e^{-q\kappa}}{\kappa}$.
    \item[(iii)] If $\underline{c} = 0$, then $P_n \rightarrow 1$. If $\underline{c} > 0$, then $P_n \rightarrow P_\infty \equiv 1 - e^{-q \kappa(\underline{c})}$.
    \end{enumerate}
\end{proposition}

The constant $\kappa$ is the limiting expected number of agents who search. It is decreasing in $\underline{c}$ and grows unboundedly as $\underline{c}$ becomes smaller. \Cref{prop:largen} has important implications for the success of the search. Plausibly, $\underline{c} > 0$, as even high-ability agents have to exert effort to find the object. Then, not even inviting an unlimited crowd will guarantee that the search is successful as $P_\infty < 1$. The reason is that\textemdash given the intense expected competition\textemdash only comparatively few agents will decide to search and the object will not be found with some probability. Yet, if a large group of agents could be invited that are partially intrinsically motivated or motivated by reputational concerns, cases with $\underline{c} = 0$ may become possible as well as the prospect that the object is found with certainty.

We now investigate the tail behavior of $P_n$. The previous section shows that the probability of success can either increase or decrease with the number of invited agents. It remains to ask whether one of these forces dominates asymptotically, if any. We say that $P_n$ is \textit{eventually increasing (decreasing)} if there exists $N$ such that $P_n$ is increasing (decreasing) for all $n>N$.

The first result is a corollary of \Cref{prop:largen}. It follows because when the lower bound of the cost distribution is zero, the probability of success converges to one and therefore cannot be eventually decreasing.

\begin{corollary} \label{cor:noeventualdecrease}
If $\underline{c}=0$, then $P_n$ cannot be eventually decreasing.
\end{corollary}

\Cref{cor:noeventualdecrease} says that when costs can be arbitrarily close to zero, the crowding-out effect cannot dominate asymptotically. The corollary rules out eventual decrease when $\underline{c}=0$, but does not establish that $P_n$ must be eventually increasing. This leaves open whether, in sufficiently large contests, the direct effect of adding more potential searchers eventually dominates the equilibrium crowding-out effect.\footnote{A famous assertion in software development is that ``Given enough eyeballs, all bugs are shallow'' \citep{raymond1999}. This is known as Linus's law, after Linus Torvalds, the creator of the Linux operating system. Using data from bug bounty programs, \citet{maillart2017} show empirically that Linus's law tends to hold.} This question is particularly subtle when $\underline{c} > 0$. In such a case, the limiting probability of success is strictly below one, so $P_n$ may approach its limit from below or from above.

The following result provides sufficient conditions for eventual increase and eventual decrease, which depend on the lower-tail behavior of $cf(c)$.

\begin{proposition} \label{prop:largePn}
The following holds:
\begin{enumerate}
    \item[(i)] If $\limsup_{c\rightarrow\underline{c}^+}cf(c)<\frac{2}{q}$, then $P_n$ is eventually increasing.
    \item[(ii)] If $\liminf_{c\rightarrow\underline{c}^+}cf(c)>\frac{2}{q}$, then $P_n$ is eventually decreasing.
\end{enumerate}
\end{proposition}

In large contests, only agents with costs very close to $\underline{c}$ search. Two possibilities are of interest. On the one hand, if the lower tail is not too concentrated, then a marginal increase in the number of invited agents adds enough potential searchers that the direct effect dominates. The condition in (i) of \Cref{prop:largePn} captures this first case. It is satisfied by many standard distributions with $\underline{c}=0$, including $F(c)=c^\alpha$ with $\alpha>0$, beta distributions, and exponential distributions.\footnote{These distributions satisfy a simple sufficient condition of bounded lower-tail elasticity. If $\limsup_{c \rightarrow 0^+} \varepsilon(c) < \infty$, then $\limsup_{c\rightarrow0^+} cf(c) = 0 < 2/q$ since $cf(c)=\varepsilon(c)F(c)$ and $F(c)\rightarrow 0$.}

The condition in (i) can also hold when $\underline{c} > 0$. An example of this case is $F(c)=\left(\frac{c-\underline{c}}{\overline{c}-\underline{c}}\right)^\alpha$ on $[\underline{c},\overline{c}]$, where $\underline{c}>0$ and $\alpha>1$. Then, as $c\rightarrow\underline{c}^+$, $f(c)=\frac{\alpha}{(\overline{c}-\underline{c})^\alpha}(c-\underline{c})^{\alpha-1} \rightarrow 0$, and so $cf(c) \rightarrow 0 < \frac{2}{q}$.

On the other hand, if the lower tail is sufficiently concentrated, then a small decrease in the equilibrium threshold excludes many agents from search. The condition in (ii) of \Cref{prop:largePn} captures this second case. Consistent with \Cref{cor:noeventualdecrease}, this condition cannot hold when $\underline{c}=0$. Indeed, if $\liminf_{c\rightarrow0^+}cf(c)>2/q$, then for all sufficiently small $c$, $f(c)$ is bounded below by $M/c$ for some $M>0$, contradicting the integrability of $f$.

When $\underline{c}>0$, the condition in (ii) can hold with, for example, $F(c)=\left(\frac{c-\underline{c}}{\overline{c}-\underline{c}}\right)^\alpha$ on $[\underline{c},\overline{c}]$, but now with $\alpha\in(0,1)$. Then, $f(c)=\frac{\alpha}{(\overline{c}-\underline{c})^\alpha}(c-\underline{c})^{\alpha-1} \rightarrow\infty$ as $c\rightarrow\underline{c}^+$ and, since $\underline{c}>0$, it follows that $cf(c)\rightarrow \infty$. Hence, the condition in (ii) holds and the loss of search incentives caused by additional competition outweighs the direct benefit of inviting more agents.

\Cref{prop:largePn} is about the tail behavior of $P_n$ and does not rule out non-monotonicity at finite values of $n$. Thus, $P_n$ can initially decrease before eventually increasing, or initially increase before eventually decreasing. We illustrate these possibilities with the following examples.


\begin{example} \label{ex:3}
Consider $F(c)=(c-1)^{50}$ for $1 \leq c \leq 2$, and let $V=4$ and $q=3/4$. \Cref{prop:largen} asserts that $c_n \rightarrow 1$ and $P_n \rightarrow 1-e^{-\frac{3}{4}\kappa}\approx 0.940$, with $\kappa=3.762$. In this case, $f(c)=50(c-1)^{49}$, so $\limsup_{c\rightarrow 1^+}cf(c)=0<8/3=\frac{2}{q}$. Therefore, $P_n$ is eventually increasing by (i)  of \Cref{prop:largePn}. \Cref{fig:example3} shows the numerical values of $c_n$ and $P_n$ for $n=10,50,100,200,1000$, for which $P_n$ initially decreases and reaches its minimum near $n=50$ before eventually increasing.
\end{example}


\begin{example} \label{ex:4}
Consider $F(c)=(c-1)^{2/3}$ for $1 \leq c \leq 2$, and let $V=4$ and $q=3/4$. \Cref{prop:largen} asserts that $c_n \rightarrow 1$ and $P_n \rightarrow 1-e^{-\frac{3}{4}\kappa}\approx 0.940$, with $\kappa=3.762$. In this case, $f(c)=\frac{2}{3}(c-1)^{-1/3}$, so $\liminf_{c\rightarrow 1^+}cf(c)=\infty>8/3=\frac{2}{q}$. Therefore, $P_n$ is eventually decreasing by (ii) of \Cref{prop:largePn}. \Cref{fig:example4} shows the numerical values of $c_n$ and $P_n$ for $n=10,50,100,200,1000$, for which $P_n$ initially increases and reaches its maximum near $n=50$ before eventually decreasing.
\end{example}


\begin{figure}[t]
\centering
\begin{subfigure}[ht]{0.5\textwidth}
\centering
\begin{tikzpicture}[scale=0.7]
\begin{axis}[
    axis y line*=left,
    ymin=1.86, ymax=1.97,
    xlabel={$\log(n)$},
    ylabel={\textcolor{red}{$c_n$}},
    ylabel style={rotate=-90}
]
\addplot[mark=asterisk,red]
    coordinates {
    (2.3,1.9603)
    (3.9,1.9296)
    (4.6,1.9169)
    (5.3,1.9045)
    (6.9,1.8764)
    };
    \label{plot_large2}
\end{axis}

\begin{axis}[
  axis y line*=right,
  axis x line=none,
  ymin=0.60, ymax=0.66,
  ylabel={\textcolor{blue}{$P_n$}},
  ylabel style={rotate=-90}
]
\addlegendimage{/pgfplots/refstyle=plot_large2}\addlegendentry{$c_n$}
\addplot[mark=o,blue]
  coordinates{
    (2.3,0.6482)
    (3.9,0.6254)
    (4.6,0.6269)
    (5.3,0.6305)
    (6.9,0.6418)
    }; 
    \addlegendentry{$P_n$}
\end{axis}
\end{tikzpicture}
\end{subfigure}
\begin{subfigure}[ht]{0.4\textwidth}
    \centering
    \vspace{-1em}
\begin{tabular}{ |c|c|c| }
 \hline
 $n$ & $c_n$ & $P_n$  \\ 
 \hline \hline
 10 & 1.9603 & 0.6482 \\ 
 \hline
 50 & 1.9296 & 0.6254 \\ 
 \hline
 100 & 1.9169 & 0.6269 \\
 \hline
 200 & 1.9045 & 0.6305 \\
 \hline
 1000 & 1.8764 & 0.6418 \\
 \hline
\end{tabular}
\end{subfigure}
\caption{$c_n$ and $P_n$ for \Cref{ex:3}.}
\label{fig:example3}
\end{figure}

\begin{figure}[t]
\centering
\begin{subfigure}[ht]{0.5\textwidth}
\centering
\begin{tikzpicture}[scale=0.7]
\begin{axis}[
    axis y line*=left,
    ymin=0.99, ymax=1.19,
    xlabel={$\log(n)$},
    ylabel={\textcolor{red}{$c_n$}},
    ylabel style={rotate=-90}
]
\addplot[mark=asterisk,red]
    coordinates {
    (2.3,1.1784)
    (3.9,1.0201)
    (4.6,1.0072)
    (5.3,1.0026)
    (6.9,1.0002)
    };
    \label{plot_large3}
\end{axis}

\begin{axis}[
  axis y line*=right,
  axis x line=none,
  ymin=0.930, ymax=0.945,
  ylabel={\textcolor{blue}{$P_n$}},
  ylabel style={rotate=-90}
]
\addlegendimage{/pgfplots/refstyle=plot_large3}\addlegendentry{$c_n$}
\addplot[mark=o,blue]
  coordinates{
    (2.3,0.9337)
    (3.9,0.9422)
    (4.6,0.9419)
    (5.3,0.9414)
    (6.9,0.9407)
    }; 
    \addlegendentry{$P_n$}
\end{axis}
\end{tikzpicture}
\end{subfigure}
\begin{subfigure}[ht]{0.4\textwidth}
    \centering
    \vspace{-1em}
\begin{tabular}{ |c|c|c| }
 \hline
 $n$ & $c_n$ & $P_n$  \\ 
 \hline \hline
 10 & 1.1784 & 0.9337 \\ 
 \hline
 50 & 1.0201 & 0.9422 \\ 
 \hline
 100 & 1.0072 & 0.9419 \\
 \hline
 200 & 1.0026 & 0.9414 \\
 \hline
 1000 & 1.0002 & 0.9407 \\
 \hline
\end{tabular}
\end{subfigure}
\caption{$c_n$ and $P_n$ for \Cref{ex:4}.}
\label{fig:example4}
\end{figure}

\section{Conclusion} \label{sec:conclusion}

We have provided a simple model to study the important dimensions along which crowdsearch can be designed. The analysis highlights the tradeoff between inviting more potential searchers and preserving individual incentives to search.

Numerous further directions can be pursued. For instance, one might introduce entry checks regarding the reputation and past achievements of agents to build a favorable pool for the search. Alternatively, would the opposite approach, allowing only newcomers, be beneficial, as this would motivate many to participate? Also, one could study the optimal design of rewards, as agents may be motivated by monetary rewards as well as by reputation gains, which could be documented by success certificates and would be valuable as an entry ticket for future programs.

\section*{Declaration of Generative AI and AI-Assisted Technologies in the Manuscript Preparation Process}

During the preparation of this work, the authors used ChatGPT (OpenAI) in the revision process to assist with language editing and checking mathematical derivations. After using this tool, the authors reviewed and edited the content as needed and take full responsibility for the content of the published article.


\singlespacing

\bibliographystyle{apalike}
\bibliography{crowdsearch}



\onehalfspacing

\appendix


\section{Mathematical Preliminaries}

This section states some facts used in the main text and the proofs. We begin with a generalization of the binomial theorem.

\begin{fact} \label{lemma:modbinom}
For $n \in \mathbb{N}$ and $x \neq 0$, $\sum_{k=0}^n \binom{n}{k} \frac{x^ky^{n-k}}{k+1} = \frac{1}{n+1}\frac{(x+y)^{n+1}-y^{n+1}}{x}.$
\end{fact}

\begin{proof}[Proof of \Cref{lemma:modbinom}] Note that $\binom{n}{k} \frac{1}{k+1} = \binom{n+1}{k+1}\frac{1}{n+1}$.
Then,
\[
\begin{aligned} 
\sum_{k=0}^n & \binom{n}{k} \frac{x^k y^{n-k}}{k+1} =  \sum_{k=0}^n \binom{n+1}{k+1} \frac{x^k y^{n-k}}{n+1} = \sum_{k=0}^{n} \binom{n+1}{k+1} \frac{x^{k+1} y^{(n+1)-(k+1)}}{(n+1)x}
\\ &= \frac{1}{(n+1)x} \left[ \sum_{k=0}^{n+1} \binom{n+1}{k} x^{k} y^{n+1-k} - y^{n+1} \right] 
= \frac{1}{n+1} \frac{(x+y)^{n+1} - y^{n+1}}{x}.
\end{aligned}
\]
\end{proof}


\begin{fact}[Bernoulli's Inequality] \label{lemma:bernoulli}
For $x \in (0,1)$ and a positive integer $m$, $(1-x)^m < (1+mx)^{-1}$.
\end{fact}

\begin{proof}[Proof of \Cref{lemma:bernoulli}]
The fact follows from the expansion of $(1-x)^{-m}$.
\end{proof}


\section{Proofs} \label{sec:proofs}

\begin{proof}[Proof of \Cref{prop:eqm}] The proof proceeds in four steps.

\noindent \textbf{Step 1.} We first argue that it is without loss of generality to restrict attention to threshold strategies. We show that if $\sigma_i^*$ is an equilibrium strategy for agent $i$, then it is non-increasing in cost. For the sake of contradiction, suppose that there exists a pair of costs $c < c'$ such that $\sigma_i^*(c) = 0$ and $\sigma_i^*(c') = 1$. Then, $c < c'$ implies
\[
\begin{aligned} 
0 &= \mathbb{E}[u_i(0,\boldsigma^*_{-i}(\boldc_{-i}),c_{i})|c_i = c] \geq \mathbb{E}[u_i(1,\boldsigma^*_{-i}(\boldc_{-i}),c_i)|c_i = c] \\ &> \mathbb{E}[u_i(1,\boldsigma^*_{-i}(\boldc_{-i}),c_i)|c_i = c'] \geq  \mathbb{E}[u_i(0,\boldsigma^*_{-i}(\boldc_{-i}),c_{i})|c_i = c'] = 0,
\end{aligned}
\]
where the first and the last inequalities follow from the definition of equilibrium. The strict inequality which follows from the definition of $u_i$ leads to a contradiction. Thus, any equilibrium strategy is a threshold strategy.

\noindent \textbf{Step 2.} Suppose that the equilibrium threshold exists and is interior, $c^* \in (\underline{c},\overline{c})$. Then, it must satisfy the following indifference condition:
\begin{equation} 
c^* = V \mathbb{E}_{\boldc_{-i}}\left[p(\boldsigma_{c^*}(\boldc_{-i})) \right]
\end{equation}
where $\boldsigma_{c^*}$ denotes a profile of threshold strategies in which all agents use the same threshold $c^*$. This condition says that an agent with cost exactly at $c^*$ must be indifferent between searching and not searching. We now express and simplify the expression for $\mathbb{E}_{\boldc_{-i}}\left[p(\boldsigma_{\hat{c}}(\boldc_{-i})) \right]$. We have

\begin{equation*}
\mathbb{E}_{\boldc_{-i}}\left[p(\boldsigma_{\hat{c}}(\boldc_{-i})) \right] = q \sum_{k=0}^{n-1} \left\{ \binom{n-1}{k} F(\hat{c})^k(1-F(\hat{c}))^{n-1-k} \left[ \sum_{t=0}^k \binom{k}{t} \frac{q^t(1-q)^{k-t}}{t+1} \right] \right\}.
\end{equation*}
Applying \Cref{lemma:modbinom} twice yields
\[
    \mathbb{E}_{\boldc_{-i}}\left[p(\boldsigma_{\hat{c}}(\boldc_{-i})) \right] = \frac{1-(1-qF(\hat{c}))^n}{nF(\hat{c})}
\]
after some algebra. Therefore, $\mathbb{E}_{\boldc_{-i}}\left[p(\boldsigma_{\hat{c}}(\boldc_{-i})) \right] = \Phi(\hat{c},q,n)$ as defined in~\eqref{eq:phi}. It follows that if an interior equilibrium threshold exists, it solves $c^* = V\Phi(c^*,q,n)$.

\noindent \textbf{Step 3.}  We next show that $c^*$ is interior and unique. First, a simple derivative shows that $\Phi(\hat{c})$ is decreasing in $\hat{c}$. Therefore, by continuity, $V\Phi(\hat{c},q,n)$ has at most one fixed point $c^*$. It is easy to see that the fixed point exists and is interior if $\underline{c} < V \Phi(\underline{c},q,n)$ and $V\Phi(\overline{c},q,n) < \overline{c}$. Because $\Phi(\underline{c},q,n) = q$ and $F(\overline{c}) = 1$, the conditions become $\underline{c} < qV$ and $V\frac{1-(1-q)^n}{n}  < \overline{c}$. The first is assumed, and the second follows from the maintained assumption $V\frac{1-(1-q)^2}{2}<\overline{c}$, since $\frac{1-(1-q)^n}{n}\leq \frac{1-(1-q)^2}{2}$ for all $n\geq2$. 

\noindent \textbf{Step 4.} For the comparative statics, we first assert some properties of $\Phi(\hat{c},q,n)$. For $\hat{c}$, it follows by direct calculation that $\partial \Phi/\partial \hat{c} < 0$. For $q$, we obtain directly: $\partial \Phi/ \partial q = (1-qF(\hat{c}))^{n-1} > 0$. For $n$, we show that if $\hat{c} > \underline{c}$,
\begin{equation}\label{eq:phin}
    \underbrace{\frac{1-(1-qF(\hat{c}))^n}{nF(\hat{c})}}_{\Phi(\hat{c},q,n)} > \underbrace{\frac{1-(1-qF(\hat{c}))^{n+1}}{(n+1)F(\hat{c})}}_{\Phi(\hat{c},q,n+1)}.
\end{equation}
Inequality~\eqref{eq:phin} simplifies to $(1-qF(\hat{c}))^n <  (1+nqF(\hat{c}))^{-1}$ which holds by \Cref{lemma:bernoulli}. The comparative statics results on $c^*(V,q,n)$ then follow from the properties of $\Phi$, since $c^*$ is the fixed point of $V\Phi$.

\end{proof}


\begin{proof}[Proof of \Cref{lemma:dpdn}]
We start with the fact that for a function $g(n)$, $\frac{\mathrm{d}}{\mathrm{d} n} (g(n))^n = (g(n))^n\left( \frac{ng'(n)}{g(n)} + \ln{g(n)}\right)$. To see why this holds, let $h(n) \equiv (g(n))^n$. Then, $\ln{h(n)} = n\ln{g(n)}$, and taking the derivative with respect to $n$ on both sides yields $\frac{h'(n)}{h(n)} = \frac{n g'(n)}{g(n)}+\ln{g(n)}$, which yields the claim after rearrangement. Applying this fact with $g(n) = 1-qF(c_n)$ yields
\begin{equation} \label{eq:Pn1}
\begin{aligned}
    \frac{\mathrm{d}P_n}{\mathrm{d} n} &= \frac{\mathrm{d}}{\mathrm{d} n} 1 - (1-qF(c_n))^n = - \frac{\mathrm{d}}{\mathrm{d} n} (1-qF(c_n))^n \\ &= (1-qF(c_n))^n \left[ \frac{nqf(c_n)}{1-qF(c_n)} \frac{\mathrm{d}c_n}{\mathrm{d}n} - \ln{(1-qF(c_n))}\right].
    \end{aligned}
\end{equation}

Now, implicit differentiation of the equilibrium condition $c_n n F(c_n) = VP_n$ yields
\begin{equation} \label{eq:Pn2}
    c_nF(c_n)+ n \frac{\mathrm{d}c_n}{\mathrm{d}n}F(c_n) + nc_nf(c_n) \frac{\mathrm{d}c_n}{\mathrm{d}n} = V \frac{\mathrm{d}P_n}{\mathrm{d}n} = \frac{c_n n F(c_n)}{P_n} \frac{\mathrm{d}P_n}{\mathrm{d}n}.
\end{equation}
Combining~\eqref{eq:Pn1} and~\eqref{eq:Pn2} yields 
\begin{equation} \label{eq:Pn3}
    \frac{\mathrm{d}P_n}{\mathrm{d}n} = (1-qF(c_n))^n \left[ \frac{nqf(c_n)}{1-qF(c_n)} \left( \frac{\frac{c_n n F(c_n)}{P_n} \frac{\mathrm{d}P_n}{\mathrm{d}n}-c_nF(c_n)}{nF(c_n) + n c_n f(c_n)}\right) - \ln{(1-qF(c_n))}\right].
\end{equation}
Solving~\eqref{eq:Pn3} for $\mathrm{d}P_n/\mathrm{d}n$, we obtain
\begin{equation} \label{eq:Pn4}
    \frac{\mathrm{d}P_n}{\mathrm{d}n} = (1-qF(c_n))^n \frac{\left[\frac{-qf(c_n)c_nF(c_n)}{(1-qF(c_n))(F(c_n)+c_nf(c_n))} - \ln{(1-qF(c_n))}\right]}{\left[1-\frac{(1-qF(c_n))^{n-1}qf(c_n)c_nnF(c_n)}{P_n(F(c_n)+c_nf(c_n))}\right]}.
\end{equation}
The denominator in~\eqref{eq:Pn4} can be written as
\begin{equation*}
\begin{aligned}
    1-&\frac{(1-qF(c_n))^{n-1}qf(c_n)c_nnF(c_n)}{P_n(F(c_n)+c_nf(c_n))} \\ &= \frac{P_n F(c_n)+ c_nf(c_n)[ P_n - (1-qF(c_n))^{n-1}q n F(c_n)]}{P_n(F(c_n)+c_nf(c_n))}.
\end{aligned}
\end{equation*}
It is strictly positive because $P_nF(c_n) > 0$ and $P_n - (1-qF(c_n))^{n-1}qnF(c_n) = 1-(1-qF(c_n))^{n-1}(1+(n-1)qF(c_n)) \geq 0$ by \Cref{lemma:bernoulli}.

Therefore, we have $\mathrm{d}P_n/\mathrm{d}n \geq 0$ if and only if the numerator term in~\eqref{eq:Pn4} is non-negative. Recall that
$L(q,y)\equiv y \frac{-(1-qy)\ln(1-qy)}{qy+(1-qy)\ln(1-qy)}$. Then
\begin{equation*}
\begin{aligned}
\frac{\mathrm{d}P_n}{\mathrm{d}n}\geq0
&\Leftrightarrow
\frac{-qf(c_n)c_nF(c_n)}
{(1-qF(c_n))(F(c_n)+c_nf(c_n))}
\geq
\ln(1-qF(c_n)) \\
&\Leftrightarrow
F(c_n)
\frac{-(1-qF(c_n))\ln(1-qF(c_n))}
{qF(c_n)+(1-qF(c_n))\ln(1-qF(c_n))}
\geq c_nf(c_n) \\
&\Leftrightarrow
L(q,F(c_n))\geq c_nf(c_n).
\end{aligned}
\end{equation*}
This is the condition stated in the lemma.
\end{proof}


\begin{proof}[Proof of \Cref{prop:initialdecline}]
We first show that $c_n$ is strictly decreasing for real $n$. For fixed $c$, from the definition of $\Phi$,
\[
    \frac{\partial \Phi(c,q,n)}{\partial n}
    =
    \frac{(1-qF(c))^n[1-n\ln(1-qF(c))]-1}{n^2F(c)}.
\]
The numerator is negative because $\frac{1}{(1-qF(c))^n} = e^{-n\ln(1-qF(c))}>1-n\ln(1-qF(c))$ by $e^x > 1+x$. Thus, $\Phi(c,q,n)$ is strictly decreasing in $n$. Since $\Phi(c,q,n)$ is also strictly decreasing in $c$, the equilibrium condition $c_n=V\Phi(c_n,q,n)$ implies that $c_n$ is strictly decreasing in $n$.

Next, by \Cref{lemma:dpdn}, $\mathrm{d}P_n/\mathrm{d}n\geq0$ if and only if $L(q,F(c_n))\geq c_nf(c_n)$. A direct calculation gives
\[
\frac{\partial L(q,y)}{\partial y}
=
-\frac{
(1-qy)^2[-\ln(1-qy)]^2
+(qy)^2[-\ln(1-qy)]
-(qy)^2
}
{\left[qy+(1-qy)\ln(1-qy)\right]^2}.
\]
Applying the inequality $\frac{\ln(1+x)}{x} \geq \frac{2}{2+x}$ for $x > -1$ to $x = -qy$, we have 
$-\ln(1-qy)\geq \frac{2qy}{2-qy}$.  The numerator is then bounded below by
\[
(1-qy)^2\left(\frac{2qy}{2-qy}\right)^2
+(qy)^2\left(\frac{2qy}{2-qy}\right)
-(qy)^2
=
\frac{(qy)^4}{(2-qy)^2}
>0.
\]
Thus, $L(q,y)$ is strictly decreasing in $y$. Since $c_n$ is strictly decreasing and $F$ is strictly increasing on its support, $F(c_n)$ is strictly decreasing in $n$. It follows that the left-hand side of~\eqref{eq:Pncondition} is strictly increasing in $n$.

By assumption, $cf(c)$ is weakly increasing in $c$. Since $c_n$ is decreasing in $n$, the right-hand side of~\eqref{eq:Pncondition}, $c_nf(c_n)$, is non-increasing in $n$. Therefore, if~\eqref{eq:Pncondition} holds at some $n\geq2$, it holds strictly for every $n'>n$. 
\end{proof}


\begin{proof}[Proof of \Cref{prop:qdagger}]
Note that $L(q,y)$ is strictly decreasing in $q$ since
\[
\frac{\partial L(q,y)}{\partial q}
=
\frac{
y^2\left[qy+\ln(1-qy)\right]
}
{\left[qy+(1-qy)\ln(1-qy)\right]^2}
<0,
\]
where the inequality follows from $-\ln(1-x) > x$ for $x\in(0,1)$, applied to $x=qy$.

Define $q^\dagger$ as the supremum of values of $q$ such that~\eqref{eq:Pncondition} holds for all possible equilibrium thresholds:
\[
q^\dagger
\equiv
\sup\left\{
q\in(0,1]:
L(q,F(c))\geq cf(c)
\text{ for all }c\in(\underline{c},\overline{c})
\right\}.
\]

First, we show that $q^\dagger>0$. Let $M\equiv \sup_c cf(c)<\infty$. Again, by $-\ln(1-x)>x$ for $x\in(0,1)$, applied to $x=qF(c)$, we have
\[
\begin{aligned}
L(q,F(c))
&=
F(c)
\frac{-(1-qF(c))\ln(1-qF(c))}
{qF(c)+(1-qF(c))\ln(1-qF(c))} \\
&=
\frac{1-qF(c)}{q}
\frac{-qF(c)\ln(1-qF(c))}
{qF(c)+(1-qF(c))\ln(1-qF(c))} \\
&>
\frac{1-qF(c)}{q} \\
&=
\frac{1}{q}-F(c)
\geq
\frac{1}{q}-1.
\end{aligned}
\]
Thus, if $q<1/(M+1)$, then $L(q,F(c))>M\geq cf(c)$ for all $c$. Hence, $q^\dagger>0$.

By the definition of $q^\dagger$ and the fact that $L(q,y)$ is strictly decreasing in $q$, if $q<q^\dagger$, then $L(q,F(c))>cf(c)$ for all $c\in(\underline{c},\overline{c})$. In particular, $L(q,F(c_n))>c_nf(c_n)$ for all $n$ and all $V$. By \Cref{lemma:dpdn}, $\mathrm{d}P_n/\mathrm{d}n>0$, so $P_n$ is increasing in $n$ for all $n\geq2$ and any $V$.

Next, suppose $q>q^\dagger$. Then there exists $c\in(\underline{c},\overline{c})$ such that $L(q,F(c))<cf(c)$. By continuity, the inequality also holds in a neighborhood of $c$. Since the equilibrium threshold varies continuously with $V$, we can choose $V$ so that, for some $n\geq2$, the equilibrium threshold $c_n$ lies in this neighborhood. Then condition~\eqref{eq:Pncondition} fails at $c_n$. Hence, by \Cref{lemma:dpdn}, $\mathrm{d}P_n/\mathrm{d}n<0$ at that value of $n$.

For $F(c)=c^\alpha$ on $[0,1]$, we have $cf(c)=\alpha c^\alpha$. Since $L(q,c^\alpha)$ is decreasing in $c$ and $\alpha c^\alpha$ is increasing in $c$, the condition is hardest to satisfy at $c=1$. Hence $q^\dagger_\alpha$ solves $L(q^\dagger_\alpha,1)=\alpha$. Since $L(q,1)$ is strictly decreasing in $q$, the solution is unique. Moreover, $q^\dagger_\alpha$ is strictly decreasing in $\alpha$, and the limits follow from $L(q,1)\rightarrow\infty$ as $q\rightarrow0$ and $L(q,1)\rightarrow0$ as $q\rightarrow1$.

\end{proof}

\begin{proof}[Proof of \Cref{prop:largen}]
\begin{itemize}
\item[(i)] The equilibrium condition can be written as $c_n n F(c_n) = V P_n$. Assume that $\lim_{n\rightarrow \infty}c_n$ is not equal to $\underline{c}$. Then, since $c_n$ is decreasing in $n$, we see that there exists some value $c'>\underline{c}$, so that $c_n>c'$ for any $n\in \mathbb{N}$. In this case, $c_n n F(c_n)>c'nF(c')$. This holds because $F$ is an increasing function, and, therefore, $F(c_n)>F(c')$. Note that $F(c')>0$, since $c'>\underline{c}$. Therefore, $\lim_{n\rightarrow \infty}{c_nnF(c_n)}\geq \lim_{n\rightarrow \infty}c'nF(c')=\infty$, which contradicts $c_nnF(c_n)=VP_n\leq V$.
\item[(ii)] Consider the case $\underline{c} = 0$. We argue by contradiction. First, assume that $\lim_{n\rightarrow \infty}nF(c_n)\neq \infty$. This implies that there exists $m\in \mathbb{N}$, such that there is an infinite sequence of natural numbers $n_1,n_2,\dots$, so that $n_iF(c_{n_i}) < m$. We have
\begin{equation} \label{eq:largenproof}
    c_{n_i} = V \frac{1-(1-qF(c_{n_i}))^{n_i}}{n_iF(c_{n_i})} \geq Vq\frac{1-e^{-qn_iF(c_{n_i})}}{qn_iF(c_{n_i})} > Vq \frac{1-e^{-qm}}{qm} > 0
\end{equation}
for any $i \in \mathbb{N}$, which contradicts (i) that $c_n$ converges to $0$. The first equality of~\eqref{eq:largenproof} follows from the equilibrium condition. The first inequality follows from the fact that $(1-x)^n \leq e^{-nx}$. The second inequality follows because the function $\frac{1-e^{-x}}{x}$ is strictly decreasing. 

The case $\underline{c} > 0$ is more involved since we cannot simply assume that $nF(c_n)$ converges. The proof requires several steps.

\textbf{Step 1.} Note that $nF(c_n)$ is bounded because $nF(c_n) = V\frac{1-(1-qF(c_{n}))^{n}}{c_n} \leq V/\underline{c}$. Moreover, $nF(c_n)$ is bounded away from zero. Otherwise, there exists a subsequence along which $nF(c_n)\rightarrow0$. Taking limits and applying L'H\^{o}pital's rule to $c_n \geq Vq \frac{1-e^{-qnF(c_n)}}{qnF(c_n)}$ then implies $\underline{c}\geq Vq$, contradicting the interiority assumption $\underline{c}<qV$.

\textbf{Step 2.} We claim that
    \begin{equation} 
        (1-qF(c_n))^n = e^{-qnF(c_n)} + o(1).
    \end{equation}
    Recall the identities $e^{-nx} = \sum_{k=0}^\infty \frac{(-1)^k (nx)^k}{k!}$ and $(1-x)^n = \sum_{k=0}^n \binom{n}{k} (-1)^k x^k$. Applying them to $\xi_n \equiv qF(c_n)$, we obtain
        \[
        \begin{aligned}
            e^{-n\xi_n} - (1-\xi_n)^n &= \sum_{k=0}^\infty (-1)^k \frac{(n\xi_n)^k}{k!} - \sum_{k=0}^n \binom{n}{k}(-1)^k\xi_n^k \\
            &= \sum_{k=0}^n (-1)^k \left[ \frac{1}{k!} - \frac{1}{n^k}\binom{n}{k} \right](n\xi_n)^k + \sum_{k=n+1}^\infty (-1)^k \frac{(n \xi_n)^k}{k!}.
                   \end{aligned}
        \]
        By Step 1, there exists $B$ such that $0 \leq n\xi_n \leq B$ for all $n$. Since $\frac{1}{k!} - \frac{1}{n^k}\binom{n}{k} \geq 0$ for $k \leq n$, applying the triangle inequality and the bound yields
\begin{equation} \label{eq:largenproofbound}
\left|e^{-n\xi_n} - (1-\xi_n)^n \right| \leq \sum_{k=0}^{n} \left[\frac{1}{k!} - \frac{1}{n^k}\binom{n}{k} \right]B^k + \sum_{k=n+1}^{\infty}\frac{B^k}{k!}.
\end{equation}

For the first term on the right-hand side of~\eqref{eq:largenproofbound}, note that, for $k \leq n$,  
\[
\frac{1}{k!} - \frac{1}{n^k}\binom{n}{k} = \frac{1}{k!} \left[ 1-\prod_{j=0}^{k-1}\left(1-\frac{j}{n}\right) \right]  \leq\frac{1}{k!} \sum_{j=0}^{k-1}\frac{j}{n}
= \frac{k(k-1)}{2n}\frac{1}{k!},
\]
where the inequality follows from the union bound. Therefore,
\[
\begin{aligned}
\sum_{k=0}^{n} \left[\frac{1}{k!} - \frac{1}{n^k}\binom{n}{k} \right]B^k & \leq \frac{1}{2n} \sum_{k=0}^{n} \frac{k(k-1)B^k}{k!} \leq \frac{1}{2n} \sum_{k=0}^{\infty} \frac{k(k-1)B^k}{k!} \\ 
& = \frac{1}{2n} \sum_{k=2}^{\infty} \frac{B^k}{(k-2)!} = \frac{B^2}{2n} \sum_{k=2}^{\infty} \frac{B^{k-2}}{(k-2)!} = \frac{B^2e^B}{2n} = O(n^{-1}).
\end{aligned}
\]
For the second term, $\sum_{k=n+1}^{\infty}\frac{B^k}{k!} = o(1)$. Hence, $\left|e^{-n\xi_n} - (1-\xi_n)^n \right| \leq O(n^{-1}) + o(1) = o(1)$ as claimed.

\textbf{Step 3.} Define $\kappa \in (0,\infty)$ as the unique solution to $\underline{c} = V \frac{1-e^{-q\kappa}}{\kappa}$. We now write
    \begin{equation} \label{eq:largenproof2}
    \begin{aligned}
        c_n - \underline{c} &= V \left[\frac{1-(1-qF(c_n))^n}{nF(c_n)} - \frac{1-e^{-q\kappa}}{\kappa}\right] \\ 
        &= Vq \left[\frac{1-e^{-qnF(c_n)}}{qnF(c_n)} - \frac{1-e^{-q\kappa}}{q\kappa} \right] +  o(1),
    \end{aligned}
    \end{equation}
    where the second equality uses Step 1, which shows that $nF(c_n)$ is bounded away from zero and infinity, and Step 2, which shows that $(1-qF(c_n))^n=e^{-qnF(c_n)}+o(1)$. Rearranging~\eqref{eq:largenproof2}, we obtain
    \[
        \frac{1}{Vq}(c_n - \underline{c}) = \left[ \frac{1-e^{-qnF(c_n)}}{qnF(c_n)} - \frac{1-e^{-q\kappa}}{q\kappa}   \right] + o(1).
    \]
    Since the left-hand side goes to zero by (i), it holds that $\frac{1-e^{-qnF(c_n)}}{qnF(c_n)} \rightarrow \frac{1-e^{-q\kappa}}{q\kappa}$. Note that $\zeta(x) = (1-e^{-x})/x$ is strictly monotone and continuous, and thus has a continuous inverse. It follows that $\frac{1}{q}\zeta^{-1}\left(\frac{1-e^{-qnF(c_n)}}{qnF(c_n)}\right) = nF(c_n) \rightarrow \kappa = \frac{1}{q}\zeta^{-1}\left( \frac{1-e^{-q\kappa}}{q\kappa}\right)$.

\item[(iii)] If $\underline{c}=0$, then by part (ii), $nF(c_n)\rightarrow\infty$. Since $(1-qF(c_n))^n\leq e^{-qnF(c_n)}$, it follows that $P_n\rightarrow1$. 

If $\underline{c}>0$, then Step 2 and part (ii) imply
\[
P_n=1-(1-qF(c_n))^n
=
1-e^{-qnF(c_n)}+o(1)
\rightarrow
1-e^{-q\kappa}.
\]
\end{itemize}
\end{proof}

\begin{proof}[Proof of \Cref{cor:noeventualdecrease}]
By (iii) of \Cref{prop:largen}, $P_n \rightarrow 1$ when $\underline{c}=0$. Moreover, $P_n<1$ for every finite $n$, since the equilibrium is interior and hence $qF(c_n)<1$. Indeed,
$P_n=1-(1-qF(c_n))^n$. Suppose, by contradiction, that $P_n$ is eventually decreasing. Then there exists $N$ such that $P_n$ is decreasing for all $n>N$. Choose any $n_0>N$. Then, for all $n>n_0$, $P_n\leq P_{n_0}<1$. Consequently, $\lim_{n\rightarrow\infty}P_n\leq P_{n_0}<1$, contradicting $P_n\rightarrow1$.
\end{proof}


\begin{proof}[Proof of \Cref{prop:largePn}]
By \Cref{prop:largen}, $c_n\rightarrow\underline{c}$, and hence $F(c_n)\rightarrow0$. Let $y_n\equiv F(c_n)$. By definition of $L$,
\[
L(q,y_n)
=
y_n
\frac{-(1-qy_n)\ln(1-qy_n)}
{qy_n+(1-qy_n)\ln(1-qy_n)}.
\]
A Taylor expansion gives
\[
\begin{aligned}
-(1-qy_n)\ln(1-qy_n)
&=
(1-qy_n)\left(qy_n+\frac{q^2y_n^2}{2}+O(y_n^3)\right) \\
&=
qy_n-\frac{q^2y_n^2}{2}+O(y_n^3),
\end{aligned}
\]
and
\[
\begin{aligned}
qy_n+(1-qy_n)\ln(1-qy_n)
&=
qy_n-(1-qy_n)\left(qy_n+\frac{q^2y_n^2}{2}+O(y_n^3)\right) \\
&=
\frac{q^2y_n^2}{2}+O(y_n^3).
\end{aligned}
\]
Therefore,
\[
L(q,y_n)
=
y_n
\frac{qy_n-\frac{q^2y_n^2}{2}+O(y_n^3)}
{\frac{q^2y_n^2}{2}+O(y_n^3)}.
\]
Thus, $L(q,F(c_n))\rightarrow 2/q$.

For part (i), suppose that
$\limsup_{c\rightarrow\underline{c}^+}cf(c)<2/q$. Since $c_n\rightarrow\underline{c}$ and $L(q,F(c_n))\rightarrow2/q$, it follows that, for all sufficiently large $n$, $c_nf(c_n)<L(q,F(c_n))$.
Thus, $P_n$ is eventually increasing by \Cref{lemma:dpdn}.

For part (ii), suppose that
$\liminf_{c\rightarrow\underline{c}^+}cf(c)>2/q$. Since $c_n\rightarrow\underline{c}$ and $L(q,F(c_n))\rightarrow2/q$, it follows that, for all sufficiently large $n$, $c_nf(c_n)>L(q,F(c_n))$. Thus, $P_n$ is eventually decreasing by \Cref{lemma:dpdn}.
\end{proof}

\end{document}